\documentclass[letterpaper, 10pt, journal]{IEEEtran}
\IEEEoverridecommandlockouts

\usepackage{cite}
\usepackage{amsmath,amssymb,amsfonts,amsthm}
\usepackage{algorithmic}
\usepackage{algorithm}
\usepackage{graphicx}
\usepackage{textcomp}
\usepackage{xcolor}
\usepackage{bm}        
\usepackage{booktabs}  
\usepackage{url}       
\usepackage{float}

\usepackage{tikz}
\newcommand{\submittedtext}{%
\footnotesize This work has been submitted to the IEEE for possible publication. Copyright may be transferred without notice, after which this version may no longer be accessible.}
\newcommand{\submittednotice}{%
\begin{tikzpicture}[remember picture,overlay]
\node[anchor=south,yshift=10pt] at (current page.south) {%
\fbox{\parbox{\dimexpr0.65\textwidth-\fboxsep-\fboxrule\relax}{\submittedtext}}};
\end{tikzpicture}%
}

\DeclareMathOperator{\blkdiag}{blkdiag}
\newtheorem{theorem}{Theorem}

\newtheorem{proposition}{Proposition}

\title{\LARGE \bf
Regularized Distributed MPC for UAV Networks: Stabilizing Coupled Motion and Hybrid Beam Alignment
}

\author{Evangelos Vlachos, \IEEEmembership{Member, IEEE}%
\thanks{Evangelos Vlachos is with the Industrial Systems Institute, ATHENA Research Center, Patras 26504, Greece (email: evlachos@athenarc.gr).}%
}

\begin{document}

\maketitle
\submittednotice
\thispagestyle{empty}
\pagestyle{empty}

\begin{abstract}
This letter investigates the coupled control problem in UAV networks utilizing high-frequency hybrid beamsteering. While phased arrays enable rapid electronic scanning, their finite Field of View (FoV) imposes a fundamental constraint that necessitates active mechanical steering of the airframe to maintain connectivity. We propose a decentralized Model Predictive Control (MPC) framework that jointly optimizes trajectory and heading to maximize network sum-capacity subject to safety constraints. Addressing the numerical instability caused by fast-fading channel nulls, we introduce a regularized surrogate cost function based on discrete spatial smoothing. We analytically prove that this approximation bounds the cost curvature, restoring the Lipschitz continuity of the gradient. Crucially, we derive a sufficient condition linking this Lipschitz constant to the controller gain, guaranteeing the contraction and linear convergence of the distributed best-response dynamics. Simulation results demonstrate that the proposed algorithm effectively navigates the trade-off between electronic beam tracking and kinematic safety, significantly systematically outperforming velocity-aligned baselines.
\end{abstract}

\begin{IEEEkeywords}
Model Predictive Control, Unmanned Aerial Vehicles, Millimeter-Wave Communication, Distributed Control, Beam Steering.
\end{IEEEkeywords}

\section{Introduction}
\label{sec:intro}

The deployment of data-intensive UAV swarms necessitates the use of high-frequency (e.g., millimeter-wave and sub-THz bands) to support multi-gigabit throughput~\cite{swindlehurst2014millimeter}. However, the high-gain directional antennas required to overcome severe path loss in these bands introduce a fundamental coupling between the agent's kinematic state (position and orientation) and the communication link quality~\cite{xiao2022survey}. Existing research typically decouples these problems. Classic formation control often assumes ideal, omnidirectional connectivity~\cite{zhang2019research}, while communication-centric placement algorithms treat agent dynamics as quasi-static constraints~\cite{li2019joint}. This separation persists despite recent surveys highlighting the critical need for integrated control-communication protocols in autonomous aerial systems~\cite{javaid2023communication}. Although Model Predictive Control (MPC) has been successfully applied to UAV collision avoidance~\cite{gu2021distributed}, standard formulations ignore the directional Field of View (FoV) of high-frequency links. In reality, narrow beamwidths create a highly non-convex optimization landscape where minor yaw or position deviations can precipitate sudden link outages. Recent joint analyses have demonstrated that this coupling directly impacts formation stability~\cite{qian2021formation} and that co-designing mobility with communication resources significantly enhances system robustness~\cite{zhou2021joint}. 

\subsubsection*{Contributions}
To address the coupling between mechanical steering and trajectory generation, we propose a decentralized MPC framework where agents actively maneuver to maintain optimal beam alignment while navigating reference paths. Our specific contributions are:

\begin{enumerate}
    \item \textit{Joint ``Fly-and-Communicate'' Framework:} We formulate a decentralized MPC scheme that jointly optimizes trajectory and mechanical heading. Unlike standard approaches that rely on velocity alignment, this formulation explicitly manages the finite Field-of-View (FoV) constraint, treating the network maximization problem as a potential game solved via Block Coordinate Descent.

    \item \textit{Theoretical Stability Analysis:} We provide a rigorous convergence guarantee for the proposed distributed solver. We prove that discrete spatial smoothing restores the Lipschitz continuity of the objective gradient (Theorem 1) and derive a closed-form contraction condition \eqref{eq:contraction_condition} that explicitly links the smoothing radius $\epsilon$ to the control system parameters to ensure linear convergence.

    \item \textit{Validation of Active Steering:} Through high-fidelity simulations, we demonstrate that the proposed coupled strategy significantly outperforms velocity-aligned baselines. We identify the specific failure modes of standard guidance during dynamic crossings and show that active mechanical yaw control is mandatory for maintaining robust connectivity in dense swarms.
\end{enumerate}

\textit{Notation:} Scalars are denoted by italic letters (e.g., $x$), vectors by bold lowercase letters (e.g., $\mathbf{x}$), and matrices by bold uppercase letters (e.g., $\mathbf{A}$). The set of real numbers is $\mathbb{R}$. For a vector $\mathbf{x}$, $\|\mathbf{x}\|$ denotes the Euclidean norm ($L_2$), and $\|\mathbf{x}\|_{\mathbf{Q}}^2 \triangleq \mathbf{x}^\top \mathbf{Q} \mathbf{x}$ denotes the weighted squared norm. The operator $\blkdiag(\cdot)$ constructs a block-diagonal matrix. The set $\mathcal{N}_i$ represents the neighbors of agent $i$.

\section{Problem Formulation}
\label{sec:problem}
\subsection{Dynamics and Constraints}

We consider a swarm of $N$ autonomous agents indexed by $i \in \{1, \dots, N\}$. The system is modeled in discrete time indexed by $k$ with sampling interval $T_s$. Let $\mathbf{x}_i(k) = [\mathbf{p}_i^\top, \mathbf{v}_i^\top, \psi_i]^\top \in \mathbb{R}^7$ and $\mathbf{u}_i(k) = [\mathbf{a}_i^\top, \omega_i]^\top \in \mathbb{R}^4$ denote the state and control input vectors, respectively, where $\mathbf{p}_i$ and $\mathbf{v}_i$ are the position and velocity in $\mathbb{R}^3$, $\mathbf{a}_i$ is the acceleration input, $\omega_i$ is the yaw rate input, and $\psi_i$ is the yaw angle (heading) of agent $i$. The discrete-time evolution of each agent is governed by the linear model:
\begin{equation}
    \mathbf{x}_i(k+1) = \mathbf{A} \mathbf{x}_i(k) + \mathbf{B} \mathbf{u}_i(k),
    \label{eq:dynamics}
\end{equation}
where $\mathbf{A}$ and $\mathbf{B}$ represent the standard decoupled double-integrator dynamics with sampling time $T_s$ \cite{franklin1998digital}. The system is subject to actuation limits and pairwise collision avoidance constraints defined as:
\begin{subequations}
\begin{align}
    & \|\mathbf{a}_i(k)\|_\infty \leq a_{\max}, \quad |\omega_i(k)| \leq \omega_{\max}, \label{eq:act_constraints} \\
    & \|\mathbf{p}_i(k) - \mathbf{p}_j(k)\| \geq d_{\min}, \quad \forall j \neq i. \label{eq:safe_constraints}
\end{align}
\end{subequations}

\subsection{Effective Channel and Beamforming Gains}
Agents communicate via directional links and operate in a propagation environment which is strongly characterized by ground reflections. We model this channel using a Two-Ray Ground Reflection model~\cite{rappaport2014millimeter}.
The complex baseband channel gain between agent $i$ and $j$ is given by:
\begin{equation}
    h_{ij} = \frac{1}{d_{\text{LoS}}} e^{-j \frac{2\pi}{\lambda} d_{\text{LoS}}} + \Gamma(\phi_g) \frac{1}{d_{\text{ref}}} e^{-j \frac{2\pi}{\lambda} d_{\text{ref}}},
\end{equation}
where $d_{\text{LoS}}$ and $d_{\text{ref}}$ are the signal propagation distances for Line-of-Sight (LoS) and ground reflection, respectively; $\lambda$ is the carrier wavelength. Here, the grazing angle $\phi_g$ is defined as the angle between the ground plane and the reflected ray path $\phi_g = \arctan\left( \frac{p_{z,i} + p_{z,j}}{d_{ij}} \right)$, where $p_{z,i}, p_{z,j}$ are the altitudes of the agents and $d_{ij}$ is the horizontal separation distance. This angle determines the reflection coefficient $\Gamma(\phi_g)$ and the phase lag of the secondary path.

\subsubsection{Spatial Smoothing Kernel}
To prevent numerical instability in gradient-based optimization due to fast-fading nulls, we formally define the spatial smoothing operator $\mathcal{S}_\epsilon$. Theoretically, for any position-dependent scalar field $f(\mathbf{p})$, this operator represents the convolution with a kernel $\phi_\epsilon$ over a support volume $\mathcal{B}_\epsilon$:
\begin{equation}
    \mathcal{S}_\epsilon [ f(\mathbf{p}_i) ] = \int_{\mathcal{B}_\epsilon} f(\mathbf{p}_i + \boldsymbol{\delta}) \phi_\epsilon(\boldsymbol{\delta}) d\boldsymbol{\delta},
    \label{eq:gain_integral}
\end{equation}
where $\boldsymbol{\delta} \in \mathbb{R}^3$ is a displacement vector and $\phi_\epsilon$ is a normalized bump function with compact support radius $\epsilon>0$. Applying this to the channel power yields the smoothed effective gain $\bar{P}_{ij} = \mathcal{S}_\epsilon [ |h_{ij}|^2 ]$, as utilized in the decoupling approximation \eqref{eq:decoupling}.

We approximate the integral in \eqref{eq:gain_integral} using a deterministic quadrature transform (sigma-point approximation) adapted from nonlinear filtering~\cite{julier1997new}. We employ a symmetric 7-point stencil aligned with the body axes:
\begin{equation}
    \mathcal{S}_\epsilon [ f(\mathbf{p}_i) ] \approx \frac{1}{M} \sum_{m=1}^M f(\mathbf{p}_i + \boldsymbol{\delta}_m),
    \label{eq:smoothing}
\end{equation}
with $M=7$. Here, $\boldsymbol{\delta}_m$ are the sigma points defined by $\boldsymbol{\delta}_m \in \{ \mathbf{0}, \pm \epsilon \mathbf{e}_x, \pm \epsilon \mathbf{e}_y, \pm \epsilon \mathbf{e}_z \}$. While this simplifies the kernel $\phi_\epsilon$ to a uniform distribution over the stencil, it effectively acts as a low-pass spatial filter, preserving the gradient regularity required by the optimizer.

\subsubsection{Directional Gain} 
The electronic beamsteering mechanism is fundamentally modeled by the Array Factor, whose functional form for a Uniform Linear Array (ULA) is mathematically equivalent to the Dirichlet kernel. To enable derivative-based control, we approximate the main lobe of the Dirichlet kernel using a differentiable Gaussian proxy~\cite{zhu2019beam}:
\begin{equation}
    G(\psi, \phi) = N_{ula} \exp\left( -\frac{(\psi - \phi)^2}{2\sigma^2} \right),
    \label{eq:gain}
\end{equation}
where $\psi$ represents the actual yaw angle (heading) of the transmitting agent, $\phi$ represents the ideal azimuth angle toward the receiving agent. The parameter $\sigma = \frac{\text{HPBW}}{2\sqrt{2 \ln 2}}$ matches the Half-Power Beamwidth of the physical array. This proxy envelops the main lobe, providing a smooth, unimodal gradient $\nabla G$ that guides the SQP solver toward optimal alignment.

\subsection{Control Objective}
We formulate the global control objective as a Finite Horizon Optimal Control Problem (FHOCP). Let $\mathbf{U}_k = \{\mathbf{u}_1(k), \dots, \mathbf{u}_N(k)\}$ denote the set of control inputs for all agents at time step $k$. The system seeks to find the optimal control sequence $\mathbf{U}^* = \{\mathbf{U}_0, \dots, \mathbf{U}_{N_c-1}\}$ over a prediction horizon $N_c$ that minimizes the following composite cost function:
\begin{equation}
    J(\mathbf{U}) = \sum_{k=0}^{N_c-1} \sum_{i=1}^N J_i\big(\mathbf{x}_i(k), \mathbf{u}_i(k), \mathbf{x}_{\mathcal{N}_i}(k)\big),
\end{equation}
where the term $\mathcal{N}_i$ represents the set of neighboring agents to the $i$-th agent, depending on the topology of the swarm. The local stage cost $J_i(\cdot)$ balances four competing objectives:
\begin{equation}
    J_i = J_{\text{track}, i} + J_{\text{reg}, i} + J_{\text{safe}, i} + J_{\text{comm}, i}
    \label{eq:local_cost_def}
\end{equation}

\subsubsection{Trajectory Tracking and Regularization}
Standard quadratic forms are employed to minimize deviation from the reference $\mathbf{p}_{\text{ref}, i}$ and penalize control effort~\cite{rawlings2017model}. The matrices $\mathbf{Q}_p$ (tracking fidelity) and $\mathbf{R}_{\text{diag}}$ (control effort) weight these competing objectives. The diagonal matrix $\mathbf{R}_{\text{diag}}= \blkdiag(\mathbf{R})$ is parameterized by the vector $\mathbf{R} = [R_{a_x}, R_{a_y}, R_{a_z}, R_\omega]^\top$. The cost is formulated as:
\begin{equation}
    J_{\text{track}, i} + J_{\text{reg}, i} = \|\mathbf{p}_i(k) - \mathbf{p}_{\text{ref}, i}(k)\|_{\mathbf{Q}_p}^2 + \|\mathbf{u}_i(k)\|_{\mathbf{R}_{\text{diag}}}^2.
    \label{eq:tracking_cost}
\end{equation}

\subsubsection{Safety Barrier Potential}
To enforce the collision avoidance constraint~\eqref{eq:safe_constraints} within the gradient-based optimization framework, we employ an interior penalty function method~\cite{nocedal2006numerical}:
\begin{equation}
    J_{\text{safe}, i} = w_{\text{safe}} \sum_{j \in \mathcal{N}_i} \frac{1}{\|\mathbf{p}_i(k) - \mathbf{p}_j(k)\|^2 - d_{\min}^2 + \mu}
    \label{eq:safety_cost}
\end{equation}
where $\mu > 0$ is a small relaxation parameter to ensure $C^2$ continuity when margin is temporarily violated.
\subsubsection{Communication Cost with Hybrid Beamsteering}
The capacity is utilized as the communication cost metric. To capture the capabilities of modern phased arrays, we employ a hybrid beamsteering model that combines mechanical yaw $\psi_i$ with electronic scanning $\phi_{ij}$ relative to the body frame. The local communication cost is:
\begin{equation}
    J_{\text{comm}, i} = - w_{\text{comm}} \sum_{j \in \mathcal{N}_i} C_{ij}
    \label{eq:comm_cost}
\end{equation}
The capacity $C_{ij}$ is calculated using an effective SNR that accounts for path loss, array gain, and aperture scan loss:
\begin{equation}
    C_{ij} = W \log_2 \left( 1 + \underbrace{\text{SNR}_{0} \cdot |h_{ij}(\mathbf{p}_i)|^2 \cdot G_{\text{hyb}}(\phi_{ij}, \psi_i)}_{\text{SNR}_{ij}(\mathbf{p}_i)} \right)
    \label{eq:capacity}
\end{equation}
where the hybrid gain function $G_{\text{hyb}}$ models the array's attempt to electronically compensate for the mechanical misalignment $\Delta \psi_{ij} = \phi_{ij} - \psi_i$. The optimal electronic steering angle $\phi_{ij}^*$ is constrained by the array's Field of View (FoV), $\Phi_{\max}$, as $\phi_{ij}^* = \min\left( \max\left( \Delta \psi_{ij}, -\Phi_{\max} \right), \Phi_{\max} \right)$.

The resulting gain is the product of the peak array factor $N_{ula}$, the scan loss $L_{\text{scan}}$ (due to projected aperture reduction), and the residual pointing error:
\begin{equation}
    G_{\text{hyb}} = N_{ula} \cdot \underbrace{(\cos \phi_{ij}^*)^{\kappa}}_{L_{\text{scan}}} \cdot \exp\left( -\frac{(\Delta \psi_{ij} - \phi_{ij}^*)^2}{2\sigma^2} \right)
    \label{eq:hybrid_gain}
\end{equation}
where $\Delta \psi_{ij}^{\text{res}} = |\Delta \psi_{ij} - \phi_{ij}^*|$ is defined as the angular difference between the mechanical misalignment and the electronic compensation; $\kappa \approx 1.3$ is the scan loss exponent.

\section{Proposed Solution}
\label{sec:solution}

We propose a decentralized MPC framework that addresses the inherent non-convexity of the joint trajectory-connectivity problem. The solution employs a BCD strategy, decomposing the global optimization into tractable local sub-problems solved via Sequential Quadratic Programming (SQP).

\subsection{Decentralized Problem Formulation}
The control scheme operates in a receding horizon fashion. At time $t$, agent $i$ solves a FHOCP to determine the optimal control sequence $\tilde{\mathbf{U}}_i$. To decouple the inter-agent constraints, we adopt a penalty method approach~\cite{xu2013block}, formulating the problem as:
\begin{align}
    \min_{\tilde{\mathbf{U}}_i} \quad & \sum_{k=0}^{N_c-1} \tilde{J}_i(\mathbf{x}_i(k), \mathbf{u}_i(k), \mathbf{x}_{\mathcal{N}_i}(k)) \label{eq:optimization_prob} \\
    \text{s.t.} \quad & \mathbf{x}_i(k+1) = \mathbf{A} \mathbf{x}_i(k) + \mathbf{B} \mathbf{u}_i(k), \\
    & \|\mathbf{a}_i(k)\|_\infty \leq a_{\max}, \quad |\omega_i(k)| \leq \omega_{\max}, \label{eq:init_constraint}
\end{align}
where $\tilde{J}_i$ is the regularized cost function described below. 

The overall cost function $J_i$ is non-convex due to the inclusion of the barrier term $J_{\text{safe}}$ and the communication term $J_{\text{comm}}$; while the directional nature of the safety barrier allows its non-convexity to be locally managed by the SQP solver, the multi-modal landscape generated by $J_{\text{comm}}$'s reliance on highly non-linear array gain and channel models prohibits SQP from guaranteeing convergence to the global capacity optimum.

\subsection{Regularized Cost Function Design}
\label{subsec:regularization}
Direct optimization of the instantaneous capacity is numerically unstable due to the deep nulls and high-frequency oscillations of the fast-fading channel. To address this, we target the spatially averaged capacity, defined as:
\begin{equation}
    \bar{C}_i = \mathcal{S}_\epsilon \left[ W \log_2(1 + \mathrm{SNR}_{ij}) \right].
    \label{eq:cap}
\end{equation}
Since the spatial smoothing operator $\mathcal{S}_\epsilon$ acts as a linear expectation (a weighted sum over the stencil), and the function $f(x) = \log_2(1+x)$ is concave, we apply Jensen's inequality.  This allows us to upper bound the expectation of the logarithm with the logarithm of the expectation:
\begin{equation}
    \underbrace{\mathcal{S}_\epsilon \left[ \log_2(1 + \mathrm{SNR}_{ij}) \right]}_{\bar{C}_i} 
    \le 
    \underbrace{\log_2 \left( 1 + \mathcal{S}_\epsilon \left[ \mathrm{SNR}_{ij} \right] \right)}_{\text{Surrogate Bound}}.
    \label{eq:jensen_bound}
\end{equation}
We adopt this smooth upper bound as our surrogate cost function, summing over all neighbors $j \in \mathcal{N}_i$:
\begin{equation}
    J_{\text{surr}, i} = \sum_{j \in \mathcal{N}_i} W \log_2\left( 1 + \mathcal{S}_\epsilon \left[ \mathrm{SNR}_{ij}(\mathbf{p}_i) \right] \right).
    \label{eq:surr_cost}
\end{equation}
To evaluate the smoothed term $\mathcal{S}_\epsilon[\mathrm{SNR}_{ij}]$, we exploit the spectral separation between the channel and the beamformer. The channel power $|h_{ij}|^2$ fluctuates at the wavelength scale $\lambda$ (small-scale fading), while the hybrid gain $G_{\mathrm{hyb}}$ varies over the much larger beamwidth scale $\sigma \gg \lambda$. This scale disparity renders the terms statistically independent over the smoothing window, allowing the spatial average of the product to be decoupled:
\begin{equation}
    \mathcal{S}_\epsilon \left[ \mathrm{SNR}_{ij}(\mathbf{p}_i) \right] 
    \approx \mathrm{SNR}_0 \cdot \underbrace{\mathcal{S}_\epsilon \left[ G_{\mathrm{hyb}} \right]}_{G_{\mathrm{hyb}}^\epsilon} \cdot \underbrace{\mathcal{S}_\epsilon \left[ |h_{ij}|^2 \right]}_{\bar{P}_{ij}}.
    \label{eq:decoupling}
\end{equation}
As proved in Section~\ref{sec:theory}, this surrogate cost possesses the Lipschitz continuity required for stable gradient-based optimization.

\subsection{SQP Solver Implementation and Complexity}
To solve the non-convex optimization problem~\eqref{eq:optimization_prob}, we employ a Sequential Quadratic Programming (SQP) scheme utilizing the \texttt{SLSQP} active-set algorithm~\cite{virtanen2020scipy}. At each major iteration $l$, the solver approximates the Lagrangian function $\mathcal{L}(\mathbf{U}, \boldsymbol{\lambda})$ with a quadratic model to generate a descent direction $\mathbf{d}_l$ via a bounded-variable least squares sub-problem.

Crucially, the gradient $\nabla J_i$ is computed utilizing the analytical expressions derived in Theorem~1, following the procedure outlined in Algorithm~\ref{alg:gradient_calc}. This ensures the solver receives precise sensitivity information regarding channel fading, allowing the Broyden–Fletcher–Goldfarb–Shanno (BFGS) algorithm~\cite{nocedal2006numerical} to maintain a positive definite Hessian approximation $\mathbf{H}_i^{(l)}$ without the computational cost of evaluating exact second-order derivatives. The control sequence is updated via $\tilde{\mathbf{U}}_i^{(l+1)} = \tilde{\mathbf{U}}_i^{(l)} + \alpha_l \mathbf{d}_l$, where the step size $\alpha_l$ is determined by a backtracking line search satisfying the Armijo condition.

\subsubsection*{Computational Analysis}
The complexity of the proposed framework is dominated by two stages: 1) Gradient Assembly: Evaluating the regularized gradient requires $M=7$ channel evaluations per neighbor for each step of the horizon $N_c$. The cost scales as $\mathcal{O}(N_c \cdot |\mathcal{N}_i| \cdot M)$, which is linear in the local neighborhood size but independent of the total swarm size $N$. 2) QP Solution: The active-set solver scales cubically with the number of decision variables $n_{\text{var}} = N_c \cdot \dim(\mathbf{u}_i)$, expressed as $\mathcal{O}(n_{\text{var}}^3)$.

\begin{figure}[!t]
\vspace{-0.2cm}
\begin{algorithm}[H] 
\caption{Regularized Gradient Computation (per Agent $i$)}
\label{alg:gradient_calc}
\begin{algorithmic}[1]
\REQUIRE State $\mathbf{x}_i$, Neighbor States $\mathbf{x}_j$, Smoothing $\epsilon$
\STATE \textbf{Output:} Regularized Gradient $\mathbf{g} = \nabla_{\mathbf{p}_i} J_{\text{surr},i}$
\FORALL{$j \in \mathcal{N}_i$}
    \STATE \textit{// 1. Generate Stencil Points (Eq. 5)}
    \STATE $\mathcal{D} \gets \{ \mathbf{p}_i, \mathbf{p}_i \pm \epsilon \mathbf{e}_x, \mathbf{p}_i \pm \epsilon \mathbf{e}_y, \mathbf{p}_i \pm \epsilon \mathbf{e}_z \}$
    \STATE \textit{// 2. Decoupled Smoothing (Eq. 20)}
    \STATE Compute $\bar{P}_{ij} \gets \frac{1}{M} \sum_{\boldsymbol{\delta} \in \mathcal{D}} |h_{ij}(\mathbf{p}_i + \boldsymbol{\delta})|^2$
    \STATE Compute $G_{\text{hyb}}^\epsilon \gets \text{SmoothBeam}(\mathbf{p}_i, \psi_i, \mathbf{p}_j; \sigma, \epsilon)$
    \STATE $\overline{\text{SNR}}_{ij} \gets \text{SNR}_0 \cdot \bar{P}_{ij} \cdot G_{\text{hyb}}^\epsilon$
    \STATE \textit{// 3. Chain Rule Accumulation (Thm. 1)}
    \STATE $w_{ij} \gets \frac{W}{\ln 2} (1 + \overline{\text{SNR}}_{ij})^{-1}$
    \STATE $\mathbf{g} \gets \mathbf{g} + w_{ij} \nabla_{\mathbf{p}_i} (\overline{\text{SNR}}_{ij})$
\ENDFOR
\end{algorithmic}
\end{algorithm}
\vspace{-0.5cm}
\end{figure}

\section{Theoretical Analysis}
\label{sec:theory}

\subsection{Regularity and Lipschitz Continuity}
\label{subsec:regularization_theory}

The convergence of the gradient-based SQP solver relies strictly on the smoothness of the objective function. The standard logarithmic capacity term $C_{ij} \propto \log(1 + \text{SNR}_{ij})$ is ill-conditioned because the gradient $\|\nabla C_{ij}\|$ becomes unbounded as $\text{SNR}_{ij} \to 0$ (fading nulls) or when the spatial derivative of the phase is discontinuous.

\begin{theorem}[Lipschitz Continuity of the Cost Gradient]
Let the spatial domain $\mathcal{X}$ be bounded. The gradient of the regularized surrogate cost $\nabla J_{\text{surr}, i}$ is Lipschitz continuous with respect to the position state $\mathbf{p}_i$. Specifically, there exists a constant $L_{\text{grad}} < \infty$ depending on the smoothing radius $\epsilon$ and beamwidth $\sigma$, such that:
\begin{equation}
    \| \nabla J_{\text{surr}, i}(\mathbf{x}) - \nabla J_{\text{surr}, i}(\mathbf{y}) \| \le L_{\text{grad}} \| \mathbf{x} - \mathbf{y} \|, \quad \forall \mathbf{x}, \mathbf{y} \in \mathcal{X}.
\end{equation}
\end{theorem}

\begin{proof}
The gradient of the cost with respect to $\mathbf{p}_i$ is derived from \eqref{eq:surr_cost} using the chain rule:
\begin{equation}
    \nabla J_{\text{surr}, i} = \frac{W}{\ln 2} \sum_{j \in \mathcal{N}_i} \underbrace{\left[\frac{1}{1 + \mathcal{S}_\epsilon[\mathrm{SNR}_{ij}]}\right]}_{\le 1} \nabla \mathcal{S}_\epsilon[\mathrm{SNR}_{ij}].
\end{equation}
The Lipschitz continuity of $\nabla J_{\text{surr}, i}$ follows from a uniform bound on its Hessian $\|\nabla^2 J_{\text{surr}, i}\|$. We utilize the decoupled approximation form derived in \eqref{eq:decoupling}, $\mathcal{S}_\epsilon[\mathrm{SNR}_{ij}] \approx \mathrm{SNR}_0\,\bar{P}_{ij}\, G_{\mathrm{hyb}}^\epsilon$. By the product rule, the Hessian of the effective SNR is:
\[
    \nabla^2 \mathcal{S}_\epsilon[\mathrm{SNR}] \propto \bar{P}_{ij} \nabla^2 G_{\mathrm{hyb}}^\epsilon + G_{\mathrm{hyb}}^\epsilon \nabla^2 \bar{P}_{ij} + 2 \nabla G_{\mathrm{hyb}}^\epsilon (\nabla \bar{P}_{ij})^\top.
\]
We analyze the curvature of the two decoupled components:

1) Hybrid Gain ($G_{\mathrm{hyb}}^\epsilon$): The raw gain $G_{\mathrm{hyb}}$ contains derivative discontinuities at the FoV boundary. However, the application of the operator $\mathcal{S}_\epsilon$ mollifies these jumps, ensuring $G_{\mathrm{hyb}}^\epsilon \in C^\infty$. The curvature is bounded by the sharper of the physical beamwidth or the smoothing kernel: $\|\nabla^2 G_{\mathrm{hyb}}^\epsilon\| = O(\epsilon^{-2} + \sigma^{-2})$.

2) Channel Power ($\bar{P}_{ij}$): The raw channel power $|h|^2$ exhibits high-frequency oscillations ($\lambda$-scale). The operator $\mathcal{S}_\epsilon$ acts as a spatial low-pass filter \eqref{eq:smoothing}, bounding the Hessian of the smoothed field by the stencil width:
\[
    \|\nabla^2 \bar{P}_{ij}\| \le C_P \epsilon^{-2}, \qquad \|\nabla \bar{P}_{ij}\| \le C_P' \epsilon^{-1}.
\]
Since $G_{\mathrm{hyb}}^\epsilon$ and $\bar{P}_{ij}$ are bounded on $\mathcal{X}$, the Hessian of the effective SNR is bounded by $\|\nabla^2 \mathcal{S}_\epsilon[\mathrm{SNR}]\| \le K_{\text{SNR}}(\epsilon^{-2} + \sigma^{-2})$. Finally, applying the chain rule to $\nabla^2 J_{\text{surr}, i}$ shows that the total curvature is a sum of bounded terms:
\[
    \|\nabla^2 J_{\text{surr}, i}\| \le K \left( \epsilon^{-2} + \sigma^{-2} \right).
\]
Consequently, the gradient is Lipschitz continuous with constant $L_{\text{grad}} = K(\epsilon^{-2} + \sigma^{-2})$.
\end{proof}

\subsection{Approximation Accuracy}
While smoothing restores regularity, it introduces a bias. We establish that this approximation error is quadratically bounded, confirming the fidelity of the surrogate model.

\begin{proposition}[Quadratic Approximation Error]
The discretization error introduced by the surrogate cost $J_{\text{surr}, i}$ relative to the continuous smoothing integral is bounded by $\mathcal{O}(\epsilon^2)$. Specifically,
\begin{equation}
    |J_{\text{cont}} - J_{\text{surr}, i}| \leq C \|\nabla^2 \mathrm{SNR}\| \epsilon^2.
\end{equation}
This quadratic convergence arises from the symmetry of the quadrature stencil ($\sum \boldsymbol{\delta}_m = \mathbf{0}$), which cancels the first-order linear terms in the Taylor expansion of the SNR field, leaving the Hessian-dependent quadratic term as the dominant residual.
\end{proposition}

\subsection{Stability and Convergence Condition}
Having established the gradient properties in Theorem 1 and Proposition 1, we derive the sufficiency condition for the convergence of the distributed BCD scheme. For the sequential fixed-point iteration to be a contraction, the global Hessian matrix $\mathcal{H}$ must be strictly diagonally dominant.

The diagonal blocks $\mathcal{H}_{ii}$, representing agent self-convexity, are dominated by the tracking cost and bounded below by $\mathcal{H}_{ii} \succeq \lambda_{\min}(\mathbf{Q}_{\text{pos}}) \mathbf{I}$. The off-diagonal blocks $\mathcal{H}_{ij}$ represent the physical coupling; as proven in Theorem 1, these are bounded by the communication gradient's Lipschitz constant $\|\mathcal{H}_{ij}\| \le w_{\text{comm}} L_{\text{grad}}$.

Applying the block-matrix Gershgorin Circle Theorem~\cite{bertsekas1989parallel}, the dynamics converge linearly if the diagonal strength exceeds the aggregate coupling. This yields the stability criterion:
\begin{equation}
    \frac{w_{\text{comm}} \, L_{\text{grad}}(\epsilon, \sigma)}{\lambda_{\min}(Q_{\text{pos}})} < 1.
    \label{eq:contraction_condition}
\end{equation}
This inequality synthesizes the theoretical framework: to ensure stability for a given $w_{\text{comm}}$, one must select a smoothing radius $\epsilon$ large enough to reduce $L_{\text{grad}}$ below the threshold defined by the position controller, while accepting the $\mathcal{O}(\epsilon^2)$ approximation error.

\section{Simulation Results}
\label{sec:results}

\begin{figure*}[t]
    \centering
    \includegraphics[clip, trim=2cm 0cm 1.2cm 1cm, scale=0.6]{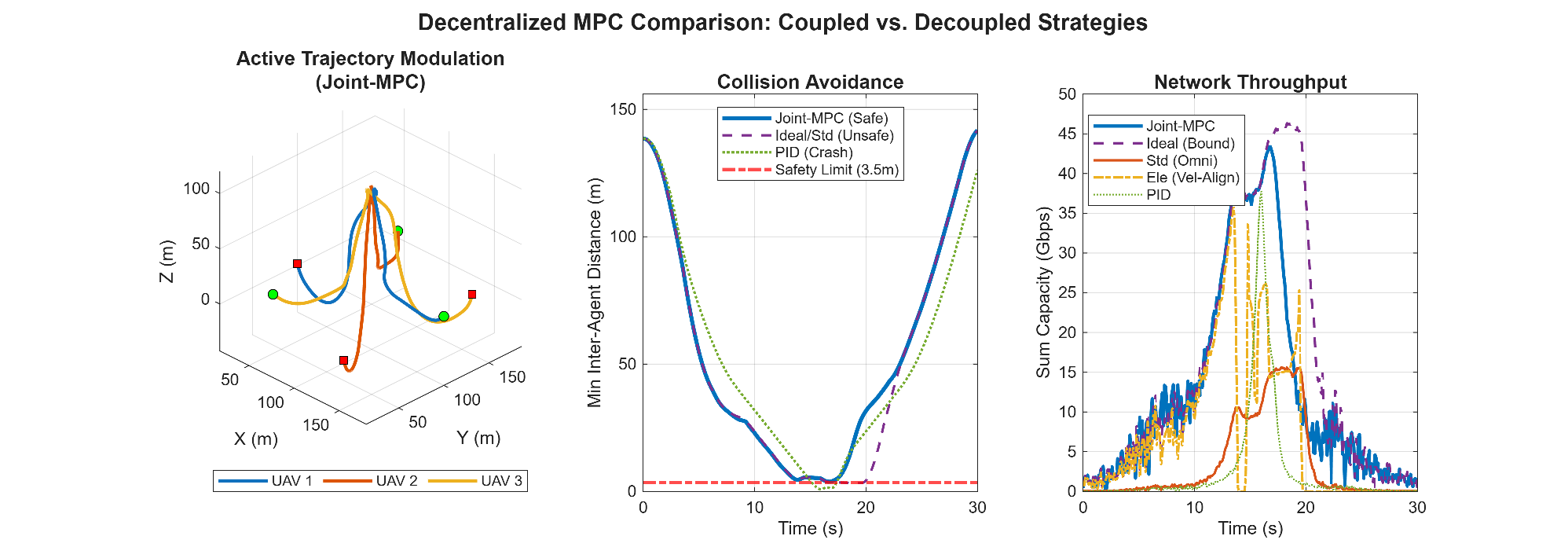}
    \vspace{-2em}
    \caption{(Left) 3D Trajectories of three UAVs performing an antipodal exchange. (Center) Evolution of minimum inter-agent distance showing strict adherence to the safety limit (dashed red line). (Right) Instantaneous Network Capacity. Note the ``notch'' in the proposed \textsf{Joint-MPC} method (blue) at $t=30$s, caused by angular velocity saturation during the close-range fly-by.}
    \label{fig:lccs_results}
    \vspace{-1em}
\end{figure*}

We evaluate the proposed framework via Monte Carlo simulations ($N_{\text{real}}=50$) in a high-fidelity 3D environment. The scenario involves a swarm of $N=3$ agents negotiating a high-speed antipodal crossing followed by a stiff formation maneuver generated via Akima splines~\cite{akima1970new}.

\subsection{System Configuration}
\label{subsec:setup}

\subsubsection{Dynamics and Control}
The swarm comprises rotary-wing UAVs modeled as double integrators ($T_s = 0.1$s) subject to actuation limits ($a_{\max} = 4.0$ m/s$^2$, $\omega_{\max} = 1.5$ rad/s). The MPC solver utilizes the \texttt{SLSQP} algorithm with a prediction horizon of $N_c = 15$ steps ($1.5$s). Cost function weights are tuned to $\mathbf{Q}_p = 2.0 \mathbf{I}$, $\mathbf{R} = [0.1, 0.1, 0.1, 0.001]^\top$, $w_{\text{comm}} = 1.0$, and $w_{\text{safe}} = 500.0$.

\subsubsection{Physical Layer}
We simulate a V-band link ($f_c = 60$ GHz) with a bandwidth of $W=2.16$ GHz, adhering to IEEE 802.11ad/ay standards~\cite{ieee80211ad}. Each agent carries a 16-element ULA (Peak Gain $\approx 12$ dBi, HPBW $6.4^\circ$). The propagation model includes geometric path loss and a ground reflection ($\Gamma=-1$).

\subsubsection{Network Topology}
For the directional UAV networks, we enforce a logical ring topology ($i \to (i+1) \pmod N$), where links require precise angular alignment of the antenna's main lobe with a specific target.

\subsection{Comparative Baselines and Metrics}
We benchmark the proposed \textsf{Joint-MPC} (joint optimization of trajectory and mechanical yaw) against four comparative strategies:
\begin{enumerate}
    \item \textsf{Ideal-MPC} (Upper Bound): Assumes the optimal kinematic path with perfect, instantaneous beam alignment ($\psi_{i} \equiv \psi_{\text{LoS}}$).
    \item \textsf{Ele-MPC} (Electronic Steering): Assumes standard velocity-aligned flight ($\psi_i = \angle \mathbf{v}_i$), leaving all misalignment compensation (within $\pm 60^\circ$ FoV) to electronic beam steering.
    \item \textsf{Std-MPC} (Omni): Standard kinematic MPC evaluated using omnidirectional antennas ($N_{ula}=1$).
    \item \textsf{PID} (Omni): A reactive potential-field controller evaluated using omnidirectional antennas.
\end{enumerate}

Performance is quantified using the following metrics averaged over time steps $T$ and agents $N$:
\begin{itemize}
    \item Global Min Distance ($\Delta_{\min}$): The minimum distance recorded between any two agents ($i \neq j$) across the entire duration ($T$). We report the Average, Variance, and Worst-Case (Min) values.
    \item Average Capacity ($\bar{C}$): The mean achievable communication rate across all active links $C_{ij}$ and time steps: $\bar{C} = \frac{1}{N T} \sum_{t, i} \sum_{j \in \mathcal{N}_i} C_{ij}(t)$.
    \item Outage Probability ($P_{\text{out}}$): The probability that the achievable link capacity $C_{ij}(t)$ drops below a critical threshold, $\mathbb{P}(C_{ij}(t) < 1 \text{ Gbps})$.
    \item Average Misalignment ($\bar{\epsilon}_\psi$): The mean absolute angular difference between the agent's actual yaw $\psi_i(t)$ and the required LoS yaw $\psi_{\text{LoS}}(t)$ over time.
    \item Control Effort ($\mathcal{E}_u$): The mean squared norm of the control input vector $\|\mathbf{u}_i(t)\|^2$ used by the agents: $\mathcal{E}_u = \frac{1}{N T} \sum_{t, i} \|\mathbf{u}_i(t)\|^2$.
\end{itemize}

\subsection{Performance Analysis}
Table \ref{tab:results} summarizes the performance metrics averaged over 50 realizations.

\subsubsection{The Failure of Velocity Alignment}
A critical finding is the failure of the \textsf{Ele-MPC} baseline, which suffers a 24\% outage probability and an average misalignment of $85.06^\circ$. While electronic steering provides some connectivity (2.16 Gbps), it is significantly less reliable than the coupled approach. This occurs because the antipodal crossing geometry requires agents to fly past each other (velocity vector forward) while communicating with a neighbor to their side (LoS vector orthogonal). Since the array is fixed to the body, the target drifts outside the electronic FoV ($\pm 60^\circ$), breaking the link. This result proves that mechanical steering (coupling yaw to position) is mandatory to keep neighbors within the electronic scan range.

\subsubsection{Safety-Capacity Trade-off}
The \textsf{Ideal-MPC} baseline achieves the highest nominal capacity (4.86 Gbps) but consistently violates the safety shell ($d_{\min}=2.54$ m $< 3.5$ m). Even the average minimum distance for the baselines (3.40 m) falls below the safety threshold. By flying in the unsafe regime, these methods artificially boost the SNR. 

In contrast, \textsf{Joint-MPC} strictly enforces safety, maintaining an average minimum distance of 4.13 m. Even in the worst-case realization, it maintains 3.03 m, far superior to the catastrophic crashes of \textsf{PID} (0.44 m) and the unsafe proximity of the standard MPCs (2.54 m). Despite this increased separation, \textsf{Joint-MPC} achieves a valid capacity of 4.00 Gbps, more than $4\times$ the performance of the blind \textsf{Std-MPC} baseline (0.98 Gbps).

\subsubsection{Alignment Precision}
The proposed framework demonstrates exceptional tracking capabilities, achieving an average misalignment of only $2.17^\circ$. This near-perfect alignment (99.07\% efficiency relative to the theoretical maximum) is achieved while simultaneously navigating complex collision avoidance maneuvers, confirming the effectiveness of the coupled optimization strategy.

\begin{table}[t]
\caption{Comparative Performance Metrics (Averaged over 50 Runs)}
\label{tab:results}
\centering
\begin{tabular}{l|c|c|c|c|c}
\toprule
\textbf{Metric} & \textbf{Joint} & \textbf{Ideal} & \textbf{Ele} & \textbf{Std} & \textbf{PID} \\
\midrule
Global Min Dist (Avg) & \textbf{4.13} & 3.40 & 3.40 & 3.40 & 0.82 \\
Global Min Dist (Var) & 0.25 & 0.14 & 0.14 & 0.14 & 0.03 \\
Global Min Dist (Min) & \textbf{3.03} & 2.54 & 2.54 & 2.54 & 0.44 \\
Avg Capacity (Gbps) & \textbf{4.00} & 4.86 & 2.16 & 0.98 & 0.79 \\
Outage Prob (\%) & 0.67 & 3.37 & 24.02 & 19.61 & 25.92 \\
Avg Misalign ($^\circ$) & 2.17 & \textbf{0.00} & 85.06 & N/A & N/A \\
Avg Effort (J) & 25.99 & 23.88 & 23.88 & 23.88 & 15.21 \\
Avg Solver Time (ms) & 84 & N/A & N/A & 49 & N/A \\
\bottomrule
\end{tabular}
\vspace{-1em}
\end{table}

\subsection{Parameter Sensitivity Analysis and Computational Feasibility}
To evaluate the impact of the smoothing radius $\epsilon$ on solver stability, we performed a sweep over $\epsilon \in [1, 20]$ cm. As detailed in Table \ref{tab:sensitivity}, the formulation demonstrates remarkable robustness. The solver success rate remains above 97\% across the range, indicating that the discrete stencil effectively regularizes the gradient even at small scales. The value $\epsilon=5$ cm balances geometric fidelity with gradient smoothness.

To assess real-time applicability, we profiled the solver on a standard consumer CPU. The \textsf{Joint-MPC} algorithm, which includes the computationally intensive 7-point spatial smoothing kernel, recorded an average execution time of $\tau_{\text{sol}} \approx 84$ ms per time step. This falls within the control sampling interval $T_s = 100$ ms, validating the feasibility of the approach even in a prototyped Python environment. Deployment on embedded hardware with C++ optimization would further reduce this latency, providing ample margin for higher-frequency control loops.

\begin{table}[t]
\caption{Sensitivity to Smoothing Radius $\epsilon$ (Averaged over 10 Runs)}
\label{tab:sensitivity}
\centering
\begin{tabular}{c|c|c|c}
\toprule
\textbf{Epsilon ($\epsilon$)} & \textbf{Success Rate} & \textbf{Avg Capacity} & \textbf{Avg Misalign} \\
(cm) & (\%) & (Gbps) & (deg) \\
\midrule
1.0 & 98.5 & 2.11 & 5.76 \\
3.0 & 98.1 & 2.03 & 6.84 \\
\textbf{5.0} & \textbf{97.7} & \textbf{2.18} & \textbf{6.16} \\
8.0 & 98.3 & 2.00 & 8.14 \\
15.0 & 98.2 & 1.94 & 8.05 \\
20.0 & 97.5 & 2.14 & 6.03 \\
\bottomrule
\end{tabular}
\vspace{-1.5em}
\end{table}

\section{Conclusion}
This letter presented a decentralized MPC framework that rigorously addresses the coupling between UAV mobility and directional networking. By introducing a spatially smoothed surrogate objective, we not only resolved the numerical instability caused by fast-fading nulls but also derived a closed-form sufficient condition for the linear convergence of the distributed dynamics. This theoretical result explicitly links the physical channel smoothing radius to the control system's stability margins. Numerical validation confirms that while standard velocity-aligned strategies suffer catastrophic link outages during dynamic maneuvers, the proposed joint control scheme maintains robust connectivity, validating the necessity of active mechanical beam steering for resilient aerial networks.

\bibliographystyle{IEEEtran}
\bibliography{literature}

\end{document}